\documentclass{IEEEtran4PSCC}
\usepackage{mathtools}
\usepackage{physics}
\usepackage{subfigure}
\ifCLASSINFOpdf
   \usepackage[pdftex]{graphicx}
\else
   \usepackage[dvips]{graphicx}
\fi
\usepackage{amsmath,amsfonts}
\usepackage{amsthm,amssymb}
\usepackage{color}

\makeatletter
\let\old@ps@headings\ps@headings
\let\old@ps@IEEEtitlepagestyle\ps@IEEEtitlepagestyle
\def\pSCCfooter#1{%
    \def\ps@headings{%
        \old@ps@headings%
        \def\@oddfoot{\strut\hfill#1\hfill\strut}%
        \def\@evenfoot{\strut\hfill#1\hfill\strut}%
    }%
    \def\ps@IEEEtitlepagestyle{%
        \old@ps@IEEEtitlepagestyle%
        \def\@oddfoot{\strut\hfill#1\hfill\strut}%
        \def\@evenfoot{\strut\hfill#1\hfill\strut}%
    }%
    \ps@headings%
}
\makeatother

\pSCCfooter{%
        \parbox{\textwidth}{\hrulefill \\ \small{23rd Power Systems Computation Conference} \hfill \begin{minipage}{0.2\textwidth}\centering \vspace*{4pt} \includegraphics[scale=0.06]{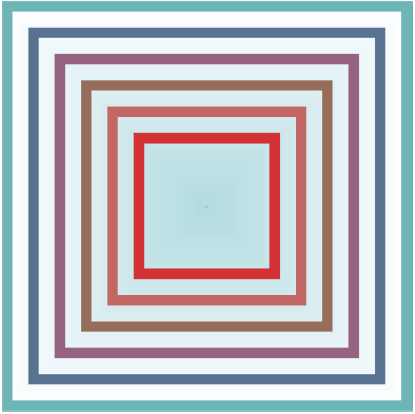}\\\small{PSCC 2024} \end{minipage} \hfill \small{Paris-Saclay, France --- June 4 -- 7, 2024}}%
}

\graphicspath{{Figs/}}

\renewcommand{\baselinestretch}{.985}

\graphicspath{{Figures/}}

\newtheorem{theorem}{Theorem}[section]
\newtheorem{proposition}[theorem]{Proposition}

\theoremstyle{definition}

\newcommand{\longthmtitle}[1]{\mbox{}{\textit{(#1):}}}

\newcommand{\complex}{\ensuremath{\mathbb{C}}}

\newcommand{\Cc}{\mathcal{C}}

\newcommand{\Ec}{\mathcal{E}}

\newcommand{\Lc}{\mathcal{L}}

\newcommand{\Nc}{\mathcal{N}}

\newcommand{\Pc}{\mathcal{P}}
\newcommand{\Qc}{\mathcal{Q}}

\newcommand\abf{\mathbf{a}}

\newcommand\dbf{\mathbf{d}}

\newcommand\pbf{\mathbf{p}}
\newcommand\qbf{\mathbf{q}}

\newcommand\vbf{\mathbf{v}}

\newcommand\ybf{\mathbf{y}}
\newcommand\zbf{\mathbf{z}}

\newcommand\Ibf{\mathbf{I}}

\newcommand\Pbf{\mathbf{P}}
\newcommand\Qbf{\mathbf{Q}}
\newcommand\Rbf{\mathbf{R}}
\newcommand\Sbf{\mathbf{S}}

\newcommand\Xbf{\mathbf{X}}
\newcommand\Ybf{\mathbf{Y}}
\newcommand\Zbf{\mathbf{Z}}

\newcommand\gammab{\boldsymbol{\gamma}}
\newcommand\xib{\boldsymbol{\xi}}

\newcommand{\ones}{\mathbf{1}}
\newcommand{\zeros}{\mathbf{0}}

\begin{document}
%
\title{Unsupervised Learning for Equitable DER Control}

\author{
\IEEEauthorblockN{Zhenyi Yuan$^1$, Guido Cavraro$^2$, Ahmed S. Zamzam$^2$, and Jorge Cort\'es$^1$}

\IEEEauthorblockA{
$^1$Department of Mechanical and Aerospace Engineering, University of California, San Diego, La Jolla, CA 92093, USA\\
$^2$Power Systems Engineering Center, National Renewable Energy Laboratory, Golden, CO 80401, USA\\
{\tt \{z7yuan,cortes\}@ucsd.edu; \{guido.cavraro,ahmed.zamzam\}@nrel.gov}}}
\maketitle
%
%
%

\begin{abstract}
In the context of managing distributed energy resources (DERs) within distribution networks (DNs), this work focuses on the task of developing local controllers. We propose an unsupervised learning framework to train functions that can closely approximate optimal power flow (OPF) solutions. The primary aim is to establish specific conditions under which these learned functions can collectively guide the network towards desired configurations asymptotically, leveraging an incremental control approach. The flexibility of the proposed methodology allows to integrate fairness-driven components into the cost function associated with the OPF problem. This addition seeks to mitigate power curtailment disparities among DERs, thereby promoting equitable power injections across the network.
To demonstrate the effectiveness of the proposed approach, power flow simulations are conducted using the IEEE 37-bus feeder. The findings not only showcase the guaranteed  system stability but also underscore its improved overall performance.
\end{abstract}

\begin{IEEEkeywords}
DER Control; Asymptotic Stability; Equitable Control; Unsupervised Learning.
\end{IEEEkeywords}

\thanksto{\noindent This work was authored in part by the National Renewable Energy Laboratory, operated by Alliance for Sustainable Energy, LLC, for the U.S. Department of Energy (DOE) under Contract No. DE-AC36-08GO28308. This work was supported by the Laboratory Directed Research and Development (LDRD) Program at NREL, and was also partially supported by NSF Award ECCS-1947050. The views expressed in the article do not necessarily represent the views of the DOE or the U.S. Government. The U.S. Government retains and the publisher, by accepting the article for publication, acknowledges that the U.S. Government retains a nonexclusive, paid-up, irrevocable, worldwide license to publish or reproduce the published form of this work, or allow others to do so, for U.S. Government purposes.}

\section{Introduction}
%
%
%
The integration of DERs has the potential to enhance power system performance and reduce greenhouse gas emissions. However, without proper regulation, the injection of power from DERs could have detrimental effects on the operation of distribution networks, such as causing significant voltage fluctuations~\cite{Mohammadi_2017_TPD}.
Fortunately, modern smart inverters are often equipped with sensing and computational capabilities, enabling DERs to learn effective control policies from data in DNs to perform ancillary services. This paper introduces a framework for synthesizing local control rules that leverage historical load data to optimize the operation of the distribution network with provable performance guarantees.

\subsubsection*{Literature Review}
Classically, generator power outputs are computed off-line by the system operator solving \emph{optimal power flow} (OPF) problems. However, the intermittency of renewable generation, the fast variability of load demands, and the uncertainty affecting DN parameters,
call for on-line closed-loop strategies which rely on measurements to update the power injections~\cite{bernstein2019real}.
Local control schemes, where decision-making occurs based on locally available information, are particularly suitable for DNs without access to real-time communication networks, as exemplified by the IEEE standard 1547~\cite{IEEE1547}. However, local schemes have intrinsic performance limitations and in general do not provide optimal performance~\cite{bolognani2019need}.
Recent advances have seen the incorporation of data-driven techniques in the design of local controllers, aiming to reduce their optimality gap compared to centralized solutions.
While most of this research has centered on Volt/Var control, i.e., the task of regulating the voltage profiles by adjusting DERs' reactive power outputs, a key consideration has emerged -- system stability.
The work~\cite{murzakhanov2023optimal} designs optimal control curves complying with the form and constraints provisioned by the IEEE 1547 Standard, ensuring stable Volt/Var dynamics. Additionally,~\cite{gupta2023deep} employs a recursive neural network whose weights correspond to the parameters of a stable reactive power control rule.
More recent efforts seek to learn stability-guaranteed nonlinear Volt/Var controllers leveraging structured neural networks, see e.g., reinforcement learning~\cite{WC-JL-BZ:22,JF-YS-GQ-SHL-AA-AW:22,Feng2023bridging} and supervised learning~\cite{yuan2023constraints,yuan2023learning} approaches.
In the context of unsupervised learning, existing results, e.g.,~\cite{gupta2021controlling}, generally do not consider system stability.

Smart inverters come equipped with the capability to regulate voltage by adjusting their reactive power injections. However, due to their limited apparent power ratings, especially when photovoltaic (PV) outputs are high, they often resort to active power curtailment for voltage regulation~\cite{XS-MM-PW-2014}. Nevertheless, the design of active power curtailment strategies, if not carefully balanced for fairness among customers, can result in an unequal distribution of burdens, particularly affecting customers located farther from the substation in terms of electrical distance~\cite{SX-YX-LC-2021, SP-MM-RS-AR:23}.
To address this challenge, various approaches, including droop-based and optimization-based methods, have been proposed to incorporate fairness considerations into active power curtailment strategies~\cite{SP-MM-RS-AR:23,SA-YW-CW-JZ-BZ-2014, PL-LA-SC-AL-GT-2019}. However, the fairness-promoting penalties proposed in the existing literature may not be suitable for designing data-driven controllers. Such controllers require penalties to be differentiable and to have the capacity to model amortized fairness, focusing on the fairness of decisions made over time rather than instantaneous fairness.

\subsubsection*{Statement of Contributions}
This paper proposes a framework for devising local power control schemes that serve as local approximations of OPF solvers.
These approximations, referred to as \emph{equilibrium functions}, map local information to active and reactive power setpoints. 
The contributions of this paper contribution are three-fold.
First, in contrast to our previous works~\cite{yuan2023constraints,yuan2023learning}, this approach adopts an unsupervised learning approach, eliminating the need for a labeled data set. This alleviates the burden of solving a large number of OPF instances off-line prior to training. Thus, the network operator can adapt the training process easily to account for changes in the network or the optimization objectives.
Second, the learned equilibrium functions depend on both the voltage magnitudes and the power injections, whereas similar approaches provides equilibrium functions that depend only on the voltages~\cite{yuan2023learning}, or that are separable in simpler components bases solely on voltage magnitudes and power injections~\cite{yuan2023constraints}.
We establish explicit conditions for these equilibrium functions to steer the network towards desirable configurations, as identified by the equilibrium functions, through an incremental algorithm controlling both active and reactive power injections.
Third, we introduce an equity-promoting penalty to facilitate learning equitable controllers defined by up to $N-1$ protected features of interest, where $N$ is the number of controllable nodes.

\subsubsection*{Notation}
$\mathbb{R}$ and $\complex$ denote the set of real and complex numbers. Matrices and column vectors are denoted by upper and lower case boldface letters, respectively. 
Calligraphic symbols denote sets. 
Given a vector $\abf$ (resp., diagonal matrix $\mathbf A$), its $n$-th (diagonal) entry is denoted by $a_n$ $(A_n)$. 
The symbol $(\cdot)^\top$ stands for transposition, $\ones$, $\zeros$ vectors of all ones and zeros, respectively, and  $\Ibf$ denotes the identity matrix with appropriate dimensions.
Operators $\Re(\cdot)$ and $\Im(\cdot)$ extract the real and imaginary parts of a complex-valued argument, and act element-wise. 
With a slight abuse of notation, the symbol $|\cdot|$ denotes the absolute value for real-valued arguments and the cardinality when the argument is a set, whereas $\|\cdot\|$ represents the $\ell_2$ vector norm and the norm that it induces on matrices. Given a matrix with real eigenvalues, $\lambda_{\max}(\cdot)$ and $\lambda_{\min}(\cdot)$ respectively denote its largest and smallest eigenvalue. 
The condition number of a matrix $\mathbf A$ is $\kappa(\mathbf A) = \|\mathbf A\| \|\mathbf A^{-1}\|$.


\section{Distribution Grid Modeling}\label{sec:modeling}
A radial single-phase (or a balanced three-phase) DN having $N+1$ buses can be modeled by a tree graph $\mathcal G=(\tilde{\mathcal N},\mathcal E)$ rooted at the substation.
The nodes in $\tilde{\mathcal N}:=\{0,\ldots,N\}$ are associated with grid buses, and the edges in $\mathcal E$ with power lines. 
The substation, labeled as node 0, is modeled as a slack bus, i.e.,  an ideal voltage source imposing the nominal voltage of 1 p.u. All the other buses are collected in $\mathcal N = \{1,\ldots,N\}$.

If we neglect the shunt admittances, the grid bus admittance matrix $\tilde \Ybf \in \complex^{(N+1)\times(N+1)}$ is defined as   
\begin{align*}
(\tilde \Ybf)_{mn} = \begin{cases}
- y_{(m,n)} & \text{ if } (m,n) \in \Ec, m \neq n, \\
0 & \text{ if } (m,n) \notin \Ec, m \neq n, \\
\sum_{k \neq n} y_{(k,n)} & \text{ if }m = n.
\end{cases}
\end{align*}
where $y_{(m,n)} \in \complex$ is the admittance of the line $(m,n)$ in~$\mathcal E$.
$\tilde \Ybf$ is symmetric and satisfies $\tilde \Ybf \ones = \zeros$. One can partition $\tilde \Ybf$ by separating the components associated with the substation and the ones associated with the other nodes to obtain 
\begin{align*}
\tilde \Ybf = \begin{bmatrix}
 y_{0}&\ybf_0^\top \\
 \ybf_0& \Ybf
 \end{bmatrix}
\end{align*}
where $y_{0} \in \complex$, $\ybf_0 \in \complex^{N}$, and $\Ybf \in \complex^{N \times N}$. The matrix $\Ybf$ is invertible if the network is connected~\cite{Kettner_2018_TPS}, which we assume; the real and the imaginary part of its inverse are denoted by $\tilde \Rbf:=\Re(\Ybf^{-1})$ and $\tilde \Xbf:=\Im(\Ybf^{-1}) \in \mathbb{R}^{N \times N}$, respectively.

The voltage magnitude at bus $n\in \Nc$ is denoted as $v_n\in \mathbb{R}$.
In general, the power injection of each bus is composed by an uncontrollable and a controllable component.
The uncontrollable part is represented by the complex load $d_n \in \complex$, whose real (imaginary) part is the active (reactive) power demand at node $n$. 
The vector $\dbf \in \complex^N$ collects the power demand for all $n\in \mathcal N$.

Assume that the buses in the subset $\Cc \subseteq \Nc$, with  $|\Cc| = C$, can control their power injections to some extent, e.g., because they host DERs or controllable loads; $p_n$ and $q_n$ denote the active and the reactive powers that they can control. 
The remaining nodes constitute the load set $\Lc:=\Nc\setminus\Cc$.
The vectors $\pbf, \qbf \in \mathbb{R}^C$ collect $\{p_n\}_{n \in \mathcal C}$ and $\{q_n\}_{n \in \mathcal C}$, respectively.  
Powers take positive (negative) values, e.g., $p_n, q_n \geq 0$ ($p_n, q_n \leq 0$), when they are \emph{injected into} (\emph{absorbed from}) the grid.

It is convenient to partition power demands and voltage magnitudes by grouping together the nodes belonging to the load and generation sets as follows,
\begin{align*}
\dbf = \begin{bmatrix}
\dbf_\Cc^\top & \dbf_\Lc^\top
\end{bmatrix}^\top, \quad
\vbf = \begin{bmatrix}
\vbf_\Cc^\top & \vbf_\Lc^\top
\end{bmatrix}^\top.
\end{align*}
Accordingly, we partition the matrices $\tilde \Rbf$ and $\tilde \Xbf$ to obtain
\begin{align*}
\tilde \Rbf = \begin{bmatrix}
\Rbf & \Rbf_\Lc \\
\Rbf_\Lc^\top & \Rbf_{\Lc\Lc}
\end{bmatrix}, \quad
\tilde \Xbf = \begin{bmatrix}
\Xbf & \Xbf_\Lc \\
\Xbf_\Lc^\top & \Xbf_{\Lc\Lc}
\end{bmatrix},
\end{align*}
with $\Rbf, \Xbf \succ 0$, see~\cite{Kettner_2018_TPS}.
For each bus $n \in \Cc$, the active and the reactive power injections are bounded by minimum and maximum values, and therefore belong to the \emph{feasible sets} 
$$\Pc_n := \{p_{n}: \underline p_{n} \leq p_{n} \leq \overline{p}_{n}\}, \;
\Qc_n := \{q_{n}: \underline q_{n} \leq q_{n} \leq \overline{q}_{n}\}.$$
Even though powers influence voltages through the well-known nonlinear power flow equations, as customary in the literature of control of DERs~\cite{murzakhanov2023optimal,gupta2023deep,WC-JL-BZ:22} for stability analysis, we adopt the linearization
\begin{equation}
\vbf(\pbf,\qbf,\dbf) = \begin{bmatrix}
    \vbf_\Cc\\
    \vbf_\Lc
\end{bmatrix} = \begin{bmatrix}
    \Rbf \\
    \Rbf_\Lc^\top
\end{bmatrix} 
\pbf + \begin{bmatrix}
    \Xbf \\
    \Xbf_\Lc^\top
\end{bmatrix} 
\qbf + \hat{\bf v}
\label{eq:v=Rp+Xq}
\end{equation}
where 
$$\hat{\bf v} = \begin{bmatrix}
\hat \vbf_\Cc^\top & \hat \vbf_\Lc^\top
\end{bmatrix}^\top :=
\tilde{\Rbf} \Re(\dbf) + 
\tilde{\Xbf} \Im(\dbf) + \ones.$$
Therefore, we can write $\vbf_\Cc$ as
\begin{align}\label{eq:vc=Rp+Xq}
    \vbf_\Cc(\pbf, \qbf, \dbf) = \Rbf \pbf + \Xbf \qbf + \hat{\vbf}_\Cc(\dbf).
\end{align}
%
%
%



\section{Problem Formulation and Proposed Approach}\label{sec:problem}
%
%
Ideally, the DER power injections should solve an \emph{optimal power flow (OPF) problem} of the form
\begin{equation}\label{eq:ORPF}
	\begin{aligned}
	(\pbf^\star,\qbf^\star) :=\arg\min_{\pbf, \qbf}\ &  f(\pbf,\qbf,\dbf) \\
	\mathrm{s.t.}\  & ~ \text{power flow equations}, \\
                        & ~\pbf \in \Pc, \qbf \in \Qc 
	\end{aligned} 
\end{equation}
where $\Pc = \times_{n\in \Cc} \Pc_n$, $\Qc = \times_{n\in \Cc} \Qc_n$ and $f$ is the cost function of interest.
Solving~\eqref{eq:ORPF} is usually difficult because of the non-convexity of the power flow equations.
%
%
Traditional approaches exploit convex relaxations, distributed optimization, and learning-based approaches, see e.g.,~\cite{bernstein2019real, Lam2012Distributed}. Such techniques rely on real-time communication among agents and the necessary supporting infrastructure is usually not present in DNs. Here instead we pursue the design of local communication-free control rules conveniently developed using neural-network-based surrogates of OPF solutions.
Precisely, historical data are used to learn, for each agent $n \in \mathcal C$, functions
\begin{align*}
\gamma_n: \mathbb{R}\times \complex \rightarrow \Pc_n, & \quad (v_n, d_n) \mapsto \gamma_n(v_n, d_n) \\
\xi_n: \mathbb{R}\times \complex \rightarrow \Qc_n, & \quad (v_n, d_n) \mapsto \xi_n(v_n, d_n) 
\end{align*}
where $\gamma_n$ ($\xi_n$) maps local voltages and power demands into approximate optimal active (reactive) power injections.
Section~\ref{sec:control} below presents
a local control scheme steering the grid toward equilibrium points 
identified by the graphs of~$\gamma_n$ and~$\xi_n$. For this reason, we refer to them as \emph{equilibrium functions}. 
We collect the~$\{\gamma_n\}_{n \in \Cc}$ and $\{\xi_n\}_{n \in \Cc}$ into the vector-valued functions 
$$\gammab: \mathbb{R}^C \times \mathbb{C}^C \mapsto \mathcal P, \qquad \xib: \mathbb{R}^C \times \mathbb{C}^C \mapsto \mathcal Q.$$ 

A common~\cite{WC-JL-BZ:22,yuan2023learning} choice of cost function $f$ in~\eqref{eq:ORPF} is a function penalizing  voltage deviations
%
%
%
%
\begin{align}\label{eq:voltage-deviation}
    f_v(\pbf,\qbf,\dbf) = \big\| {\bf v} (\pbf,\qbf,\dbf) - \boldsymbol{1} \big\|^2.
\end{align}
%
%
The dependence of $\vbf$ on $\pbf,\qbf$, and $\dbf$ can be approximated using~\eqref{eq:v=Rp+Xq} making~\eqref{eq:voltage-deviation} a linear least-squares problem.
Hence, the differentiation process required for learning $\gammab$ and $\xib$ becomes simple to carry out. Other formulations of the loss functions may consider the electric losses in the network, or the deviation from a pre-determined substation power trajectory.

\subsubsection*{Equitable Controllers Design}
The proposed learning-based control approach determines both the active and reactive power injections of the inverters for voltage regulation within a distribution network. In scenarios with significant penetration of renewable energy sources, control strategies tend to curtail power injections from DERs to regulate voltage levels. The frequent implementation of such control actions at specific locations can necessitate the installation of energy storage systems to manage the frequent curtailment of renewable generation. Studies in the literature, e.g.,~\cite{SX-YX-LC-2021}, have shown active power curtailments tend to increase with electrical distance to the substation.

Therefore, it becomes imperative to address the fairness and equity aspects of the proposed learning-based controllers. To achieve this, consider a protected feature, denoted as $z_n$, associated with the DER located at node~$n$. This feature represents any characteristic which we want the control decisions to become independent of, for instance, the electrical distance to the substation. Note that merely excluding the use of $z_n$ in the learning algorithm is insufficient to ensure this independence, as the controller may inadvertently learn to rely on other features that are correlated with the protected feature. Thus, in the proposed learning approach, we introduce a penalty term to the learning loss. This penalty term encourages control decisions to be independent of the protected feature(s), ensuring fairness in the control strategy. 

We normalize all $z_n$ such that $\|z_n\| = 1$ for $n$ in $\Cc$, and let ${\bf z}_{\Cc}$ collect all $z_n$ for $n$ in $\Cc$. We define the equity-promoting penalty as
\begin{align}\label{eq:equity-penalty}
    {f}_{eq}(\pbf, \qbf, \dbf) =  \left\| \langle \boldsymbol{\gamma}(\vbf_{\Cc}, {\bf d}_{\Cc}), {\bf z}_{\Cc} \rangle\right\| .
\end{align}
This penalty promotes the decision vectors to be orthogonal on the protected features vector, promoting that decisions are independent of the protected feature when the penalty is added to the voltage regulation loss in~\eqref{eq:voltage-deviation}.

It is worth highlighting that the design of this penalization approach allows for the incorporation of not just a single protected feature but multiple ones. Specifically, we may include up to $N-1$ protected features, where the aim is for the control approach not to depend on them. This limit comes from the fact that the penalty will encourage the control policy to reside within the null space of the protected features vectors. This null space in general vanishes when considering $N$ protected features or more. It is also worth noting that with the increase of protected features, the control policy will be restricted to a smaller space. Hence, it is expected that with the inclusion of more protected features, the performance of the proposed controller will be affected in terms of the voltage regulation loss~$f_v$.

\section{Stable Incremental Control Strategy}\label{sec:control}

We propose to control the DER hosted by bus $n$ according to the \emph{incremental} strategy
\begin{equation}\label{eq:power_update}
\begin{bmatrix}
p_n(t+1) \\ q_n(t+1)
\end{bmatrix} = (1 - \epsilon)
\begin{bmatrix}
p_n(t) \\ q_n(t)
\end{bmatrix} +
\epsilon
\begin{bmatrix}
\gamma_n(v_n(t),d_n(t)) \\ \xi_n(v_n(t),d_n(t))
\end{bmatrix}
\end{equation}
where $\epsilon \in [0,1]$.
Note that, if $p_{n}(0)  \in \Pc_n$ and $\ q_{n}(0) \in \Qc_n$, then $p_{n}(t)$ and $q_{n}(t)$ are feasible for every $t \geq 0$.
Let $p_n^\star,q_n^\star,v_n^\star$ be the power injections and the voltage magnitude at an equilibrium of~\eqref{eq:power_update}; it holds that 
\begin{align}\label{eq:eq_point}
\begin{bmatrix}
p_{n}^\star \\ q_{n}^\star
\end{bmatrix} = 
\begin{bmatrix}
\gamma_n(v_n^\star, d_n) \\ \xi_n(v_n^\star, d_n)
\end{bmatrix}
\end{align}
i.e., equilibria are exactly identified  by the equilibrium functions $\gamma_n(v_n, d_n)$ and $\xi_n(v_n, d_n)$.

%
Inspired by the solutions proposed in the literature, see~\cite{kashani2018smart,seuss2016advanced,turitsyn2011options,HZ-HJL:15}, or even adopted by utilities, see~\cite{IEEE1547,emmanuel2020estimation}, we consider equilibrium functions that are non-increasing in~$v_n$. 
Before stating the main result characterizing system stability, we define the following auxiliary variables

\vspace*{-2ex}
\begin{small}
$$L_p := \max_n \max_{v_n} \left|\frac{\partial \gamma_n (v_n, d_n)}{\partial v_n}\right|, L_q := \max_n \max_{v_n} \left| \frac{\partial \xi_n(v_n, d_n)}{\partial v_n}\right|$$
\end{small}

\noindent which represent the maximum slope absolute values across equilibrium functions.
Second, define the scalar $\alpha^*$ as the solution of the following optimization problem
\begin{equation}
\alpha^* := \arg \min_{\alpha \geq 0} \|\Xbf - \alpha \Rbf\|^2
\label{eq:alpha_prob}
\end{equation}
and the matrix $\hat \Xbf := \Xbf - \alpha^* \Rbf$. Note that~\eqref{eq:alpha_prob} is strictly convex and hence $\alpha^*$ is well defined.

The next result, proved in the Appendix~\ref{app:proofs}, identifies sufficient conditions on the equilibrium functions that ensure the stability of~\eqref{eq:dyn_sys}. 
We will assume that load demands are fixed in time,
$$d_n(t) = d_n, \; n\in \mathcal N, \; t\geq 0.$$
This is motivated by the fact that we envision the control algorithm~\eqref{eq:power_update} acting on a time scale that is faster than the load variability. 
%
%
Using~\eqref{eq:v=Rp+Xq} and~\eqref{eq:power_update} yields  the dynamical system
\begin{subequations}\label{eq:dyn_sys}
\begin{align}
    \begin{bmatrix}
        \pbf(t+1) \\ \qbf(t+1)
    \end{bmatrix} & = (1-\epsilon)
    \begin{bmatrix}
        \pbf(t) \\ \qbf(t)
    \end{bmatrix} + \epsilon 
    \begin{bmatrix}
        \gammab(\vbf_\Cc(t),\dbf_\Cc) \\ 
        \xib(\vbf_\Cc(t),\dbf_\Cc)
    \end{bmatrix}   \label{eq:dyn_sys:pwr_upd}\\
\vbf_\Cc(t) &= 
\Rbf  \pbf(t) + 
\Xbf  \qbf(t) + \hat \vbf_\Cc.
\end{align}    
\end{subequations}

\begin{theorem}\longthmtitle{Global asymptotic stability of incremental control strategy}\label{thm:NIF_stable}
The system~\eqref{eq:dyn_sys} has an unique and globally asymptotically stable equilibrium point if
\begin{subequations}
\begin{align}
&\frac{\partial \gamma_n (v_n, d_n)}{\partial v_n} \leq 0, \ \frac{\partial \xi_n (v_n, d_n)}{\partial v_n} \leq 0,~n \in \Cc,
\label{eq:nonincr_cond}\\
& L_q < \frac{1}{\kappa(\Rbf^{\frac 1 2}) \|\hat\Xbf\|}, \label{eq:nonincr_cond_a} \quad \text{and}\\
& \epsilon < \min \Big\{1, \frac{2}{1 \!+\! \kappa(\Rbf^{\frac 1 2}) L_q \|\hat\Xbf\| \!+\!  (L_p \!+\! \alpha^* L_q) \|\Rbf\|} \Big\}.
\label{eq:nonincr_cond_b}
\end{align}
\end{subequations}
\end{theorem}
Note that introduction of $\alpha^*$ in~\eqref{eq:alpha_prob} makes the condition on $L_q$ in~\eqref{eq:nonincr_cond_a} easier to satisfy. Another key feature of the above design is that we require $\{\gamma_n\}_{n \in \Cc}$ and $\{\xi_n\}_{n \in \Cc}$ to be non-increasing functions of $v_n$, and thus in the following we refer to this as \emph{non-increasing function} (NIF) design. 
Besides the stability consideration, the intuition behind the NIF design also comes from the fact that voltage magnitudes are in general increasing functions of the active and reactive powers, cf.~\eqref{eq:vc=Rp+Xq}. In an effort to bring voltages as close as possible to the nominal one, see~\eqref{eq:voltage-deviation}, generators should inject more (less) power when the voltage is lower (higher) primarily to support system voltage levels, which is consistent with the behavior of NIF design.

\section{Unsupervised Learning of Equilibrium Functions}
\label{sec:learning}
Our goal here is to learn the equilibrium functions $\{\gamma_n\}_{n \in \Cc}$ and $\{\xi_n\}_{n \in \Cc}$ by training neural networks. The training is performed in an off-line fashion. Let $\theta_n$ ($\phi_n$) be the parameters of the neural network used to learn $\gamma_n(\cdot)$ ($\xi_n(\cdot)$) and collect them in the vectors $\boldsymbol{\theta}$ ($\boldsymbol{\phi}$). 
In the following, we explicitly denote the dependence of the equilibrium functions from the neural network parameters as $\boldsymbol{\gamma}({\bf v}_\Cc, {\bf d}_\Cc; \boldsymbol{\theta}), \boldsymbol{\xi}({\bf v}_\Cc, {\bf d}_\Cc; \boldsymbol{\phi})$

The equilibrium function design can be formulated as the task of selecting $\boldsymbol{\theta}$ and $\boldsymbol{\phi}$ to minimize an operational cost such as regulating voltage. To address varying conditions of the system, we sample diverse load conditions $\{{\bf d}^{(m)}\}_{m=1}^M$ and their corresponding voltage magnitudes $\{{\bf v}^{(m)}\}_{m=1}^{M}$. With a slight abuse of notation, we write the losses evaluated for the $m$-th sample as $f_v(\vbf^{(m)}, \dbf^{(m)})$ (or $f_v^{(m)}$) and $f_{eq}(\vbf^{(m)}, \dbf^{(m)})$ (or $f_{eq}^{(m)}$) given their dependence on the voltage magnitudes and the load demand through the $\boldsymbol{\gamma}$ and $\boldsymbol{\xi}$. Then, the samples are used obtain the controllers design as follows
\begin{equation}\label{eq:controller-design-1}
    \begin{aligned}
    (\boldsymbol{\theta}^\star, \boldsymbol{\phi}^\star) \in \arg\min_{\boldsymbol{\theta}\in\tilde{\boldsymbol{\Theta}},\
    \boldsymbol{\phi}\in\tilde{\boldsymbol{\Phi}}}\quad \frac{1}{M} \sum_{m=1}^{M} &f_v(\vbf^{(m)}, \dbf^{(m)}) + \\ \lambda &f_{eq}(\vbf^{(m)}, \dbf^{(m)}) .
\end{aligned}
\end{equation}
where 
$\lambda$ is a tunable penalty parameter and the loss functions for the $m$-th sample are given by
\begin{align*}
&f_v^{(m)} = \big\| \vbf \big(\boldsymbol{\gamma}({\bf v}_\Cc^{(m)}, {\bf d}_\Cc^{(m)}; \boldsymbol{\theta}),\boldsymbol{\xi}({\bf v}_\Cc^{(m)}, {\bf d}_\Cc^{(m)}; \boldsymbol{\phi}), \dbf \big) - \boldsymbol{1} \big\|^2, \\
&f_{eq}^{(m)} = \big\| \langle \boldsymbol{\gamma}({\bf v}_\Cc^{(m)}, {\bf d}_\Cc^{(m)}; \boldsymbol{\theta}), \zbf_\Cc \rangle\big\|.
\end{align*}
The feasibility sets $\tilde{\boldsymbol{\Theta}}$ and $\tilde{\boldsymbol{\Phi}}$ are defined in the space of the neural network parameters. We identify them later to ensure that the conditions required for the stability of the control strategy are satisfied, 
After $\boldsymbol{\theta}^\star$ and $\boldsymbol{\phi}^\star$ are identified, the local control roles are deployed at all DERs. 

Problem~\eqref{eq:controller-design-1} can be tackled using a stochastic projected gradient descent (SPGD) approach as follows:
\begin{align*}
    \boldsymbol{\theta}^{k+1} &= \bigg[\boldsymbol{\theta}^k - \frac{\delta}{2B} \grad_{\boldsymbol{\theta}} \big( \sum_{m\in\mathbb{B}} f_v^{(m)} + \lambda f_{eq}^{(m)} \big) \bigg]_{\tilde{\boldsymbol{\Theta}}}
    \\
    \boldsymbol{\phi}^{k+1} &= \bigg[\boldsymbol{\phi}^k - \frac{\delta}{2B} \grad_{\boldsymbol{\phi}} \big( \sum_{m\in\mathbb{B}}f_v^{(m)} + \lambda f_{eq}^{(m)} \big) \bigg]_{\tilde{\boldsymbol{\Phi}}}
\end{align*}
where $\delta$ denotes the step size (learning rate), the set $\mathbb{B}$ is a batch of $B$ scenarios,
%
%
and the operators $[\cdot]_{\tilde{\boldsymbol{\Theta}}}$ and $[\cdot]_{\tilde{\boldsymbol{\Phi}}}$ denote the projection operator over $\tilde{\boldsymbol{\Theta}}$ and $\tilde{\boldsymbol{\Phi}}$, respectively.

This approach circumvents the need to produce training samples, which require solving many instances of the optimization problem offline before training the controllers. In addition, network operators can easily adapt their training process online to account for changes in the network or to modify the optimization objective.
%
%
On the other hand, this approach requires to differentiate through the network model during training, which can be done easily when a linearized power flow model is used. Note that if the model is not available to differentiate through, then the problem can be formulated as a reinforcement learning policy optimization problem, something which we do not pursue here.

One key aspect of the proposed controllers is the stability-ensuring constraints~\eqref{eq:nonincr_cond} and \eqref{eq:nonincr_cond_a} on the equilibrium functions to enable their safe integration within distribution networks. These stability constraints  translate into constraints on the neural network parameters. Consequently, we design the neural networks in a way that makes the projection process easy to implement while not restricting their representation capabilities. For the controller at bus $n\in\mathcal{C}$, we use a single-layer neural network
%
%
to map the inputs $(v_n, d_n)$ to the control actions $(p_n, q_n)$ as
\begin{equation}
    \begin{aligned}
    \begin{bmatrix}
        p_n\\
        q_n
    \end{bmatrix} = \sum_{h=1}^{H} \begin{bmatrix}
        w_{p,n}^{(h)}\\
        w_{q,n}^{(h)}
    \end{bmatrix} \text{tanh} \Big( &a_n^{(h)} v_n +   b_n^{(h)} p_{L,n} +\\ &c_{n}^{(h)} q_{L, n} +        d_n^{(h)} \Big) + \begin{bmatrix}
        e_{p,n}\\
        e_{q,n}
    \end{bmatrix}
    \end{aligned}
\end{equation}
where we separate $d_n$ into its real and imaginary parts $p_{L, n}$ and $q_{L, n}$. The neural network parameters for the controller at bus $n$ are $w_{p,n}$, $w_{q,n}$, $a_n^{(h)}$, $b_n^{(h)}$, $c_n^{(h)}$, $d_n^{(h)}$, $e_{p,n}^{(h)}$, and $e_{q,n}^{(h)}$ for all value of $h\in\{1, \ldots, H\}$. We choose the `tanh$(\cdot)$' activation function because, as we show later, it simplifies our projections to satisfy neural network parameter constraints.
%
%

The next result provides sufficient conditions for the neural network parameters to ensure that the requirements~\eqref{eq:nonincr_cond} and~\eqref{eq:nonincr_cond_a} are met. The value of the stepsize parameter $\epsilon$ can then be set ex post according to~\eqref{eq:nonincr_cond_b}.
%

\begin{proposition}\longthmtitle{Conditions on neural network parameters}
\label{prop:NNparameters}
Conditions \eqref{eq:nonincr_cond} and \eqref{eq:nonincr_cond_a} are satisfied if the neural network parameters are such that
\begin{subequations}
\begin{align}
& a_n^{(h)} \geq 0, w_{p,n}^{(h)}, w_{q, n}^{(h)} \leq 0, \ \quad \forall~h=\{1,...,H\},
\label{eq:nn-conditions-1}
\\
& \sum_{h=1}^H |w_{q,n}^{(h)} a_{n}^{(h)}| \leq \frac{1}{\kappa(\Rbf^{\frac{1}{2}}) \|\hat{\Xbf}\|},\quad \forall~n \in \Cc.\label{eq:nn-conditions-2}
\end{align}
\end{subequations}
\end{proposition}
\begin{proof}
The derivatives of $\gamma_n$ and $\xi_n$ with respect to $v_n$ are
\begin{align*}
    \frac{\partial \gamma_n}{\partial v_n} \!=\! \sum_{h=1}^{H} w_{p,n}^{(h)} a_n^{(h)} \text{sech}^2\bigg( a_n^{(h)} v_n \!+\!   b_n^{(h)} p_{L,n} \!+\! c_{n}^{(h)} q_{L, n} \!+\!        d_n^{(h)} \bigg)\\
    \frac{\partial \xi_n}{\partial v_n} \!=\! \sum_{h=1}^{H} w_{q,n}^{(h)} a_n^{(h)} \text{sech}^2\bigg( a_n^{(h)} v_n \!+\!   b_n^{(h)} p_{L,n} \!+\! c_{n}^{(h)} q_{L, n} \!+\!        d_n^{(h)} \bigg)
\end{align*}
and, under~\eqref{eq:nn-conditions-1}, are non-positive,. Thus~\eqref{eq:nonincr_cond} holds.

Condition~\eqref{eq:nn-conditions-2} then implies~\eqref{eq:nonincr_cond_a} using the fact that $\text{sech}^2(\cdot) \leq 1$ and applying the triangle inequality followed by the Cauchy–Schwarz inequality.
%
%
\end{proof}
Proposition~\ref{prop:NNparameters} characterizes the feasible sets $\tilde{\boldsymbol{\Theta}}$ and $\tilde{\boldsymbol{\Phi}}$. In practical implementations, one can fix the $a_n^{(h)}$'s and find the $w_{q,n}^{(h)}$'s to satisfy~\eqref{eq:nn-conditions-2} using a simple bisection approach or by solving the projected optimization problem using convex solvers~\cite{diamond2016cvxpy}. This can be interpreted as decreasing the absolute value of the $w_{q,n}^{(h)}$'s whenever the condition is violated. The restriction on the sign of these parameters as a result of~\eqref{eq:nn-conditions-1} further simplifies the projection. The details of these formulations are omitted here due to space limitations.

\section{Case Study}
We conduct case studies on a single-phase equivalent of the IEEE 37-bus feeder, see Fig.~\ref{fig:ieee37}. 
Five DERs are deployed at buses $\Cc = \{10,15,16,20,25\}$, with generation capability $\bar p_n  = 0.4 $ MW, $\underline p_{n} = 0$ MW, and $\bar q_{n} = -\underline q_{n} = 0.4 $ MVAR for all $n\in \Cc$.
For our experiments, we use the data set synthesized in~\cite{yuan2023learning}, which consists of $B=1440$ minute-based load and uncontrollable solar generation scenarios. 

\begin{figure}[tb]
\centering	
\includegraphics[width=0.5\columnwidth]{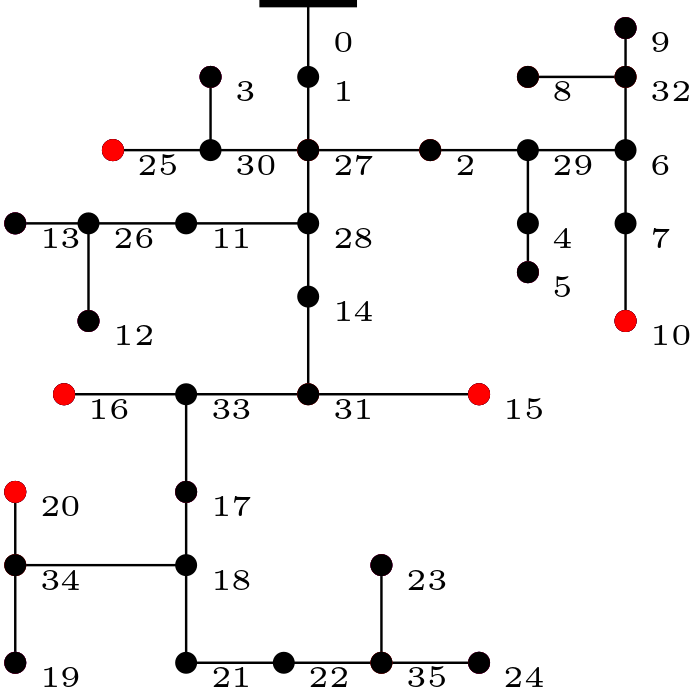}
\caption{IEEE 37-bus feeder. Red nodes represent buses hosting DERs, black nodes represent loads.}
\label{fig:ieee37}
\end{figure}


We implement the neural network approach according to Proposition~\ref{prop:NNparameters} using Pytorch and conduct the training process in Google Colab with a single TPU with 32 GB memory. The number of episodes and the number of neurons $H$ are 5000 and 50, respectively, and the neural networks are trained using the Adam optimizer~\cite{kingma2014adam} with learning rate $\delta = 0.01$.
For convenience, we set all $a_n^{(h)}$ in~\eqref{eq:nn-conditions-1} to be 1 for all $h \in \{1,...,H\}$ and $n \in \Cc$. We solve a convex program to project the neural network parameters to satisfy~\eqref{eq:nn-conditions-2} every 10 epochs using CVXPY~\cite{diamond2016cvxpy}.


To evaluate its effectiveness, our NIF design is compared with the linear Volt/Watt~\cite{braslavsky2017voltage} and Volt/Var~\cite{turitsyn2011options} controller:
\begin{align*}
&p_n(t+1) = \rho_n(v_n(t))  \\
&\rho_n(v_n) := \begin{cases}
p_{\max,n} & v_n(t) \leq v_{\min,n}^{\rm th}, \\
p_{\min,n} & v_n(t) \geq v_{\max,n}, \\
-b_n(v_n(t) \!-\! v_{\min,n}^{\rm th}) \!+\! p_{\max,n} & \text{otherwise},
\end{cases}
\end{align*}
\vspace{-2ex}
\begin{align*}
&q_n(t+1) = \varrho_n(v_n(t))  \\
&\varrho_n(v_n) := \begin{cases}
q_{\max,n} & v_n(t) \leq v_{\min,n}, \\
q_{\min,n} & v_n(t) \geq v_{\max,n}, \\
-c_n(v_n(t) \!-\! v_{\min,n}) \!+\! q_{\max,n} & \text{otherwise},
\end{cases}
\end{align*}
where $b_n = \frac{p_{\max,n}-p_{\min,n}}{v_{\max,n}-v_{\min,n}^{\rm th}}$, $c_n = \frac{q_{\max,n}-q_{\min,n}}{v_{\max,n}-v_{\min,n}}$, and $v_{\min,n}^{\rm th} = 1.03$, $v_{\min,n} = 0.95$ and $v_{\max,n} = 1.05$ for all $n \in \Cc$.

We consider the vector $\zbf_\Cc$ in~\eqref{eq:equity-penalty} to be the diagonal of $\Rbf$, which measures the electric distance between the buses hosting the controllable DERs and the substation. This consideration is based on observations in the literature that the amount of curtailed renewable energy tends to increase with the electrical distance to the substation~\cite{SX-YX-LC-2021, SP-MM-RS-AR:23}.

With $\lambda = 0.0154$ in~\eqref{eq:controller-design-1}, Fig.~\ref{fig:curves} shows the learned Volt/Watt and Volt/Var curves for node 25 with the 795-th minute-based load and generation profile under the NIF design. It can be seen that the equity-promoting penalization~\eqref{eq:equity-penalty} promotes a slightly different shape of the learned curves.

\begin{figure}[tb]
    \centering
    \subfigure[Volt/Watt Function]{
    \includegraphics[width=.24\textwidth]{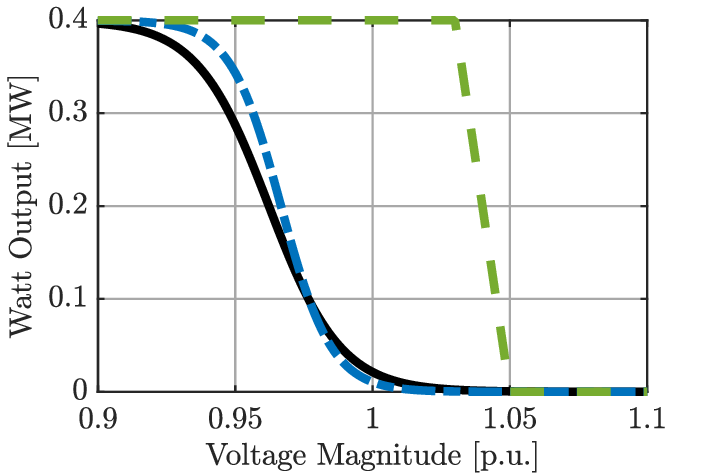}
    \label{fig:p-curve}
    }\hspace{-3ex}
    \subfigure[Volt/Var Function]{
    \includegraphics[width=.24\textwidth]{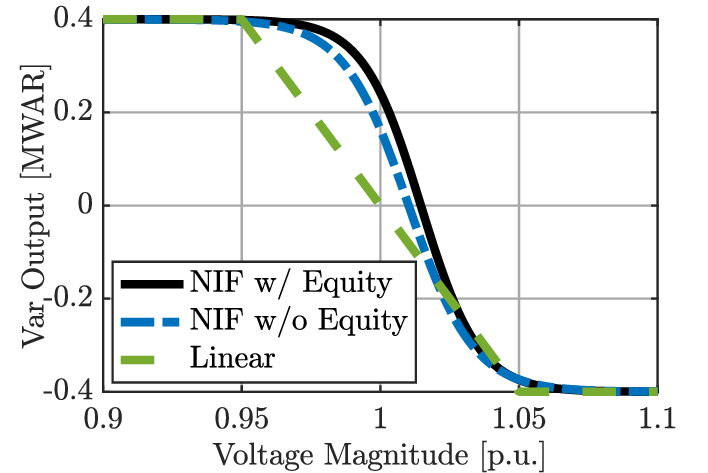}
    \label{fig:q-curve}
    }
  \caption{Learned (a) Volt/Watt and (b) Volt/Var curves for node 25 using 795-th minute-based load and generation profiles with and without equity design.}
  \label{fig:curves}
\end{figure}

Next, we verify the stability properties of the NIF design. Consider the 1095-th minute load-generation profile and 100 iterations of~\eqref{eq:power_update}, the active and reactive power setpoint trajectories converge to their final values if $\epsilon = 0.1$, which satisfies~\eqref{eq:nonincr_cond_b}, cf. Fig.~\ref{fig:stability-0.1}, while the case $\epsilon = 1$ fails, cf. Fig.~\ref{fig:stability-1}. This is consistent with our conclusion in Theorem~\ref{thm:NIF_stable}.


We then perform simulations to evaluate the control performance in case the load-generation profiles change with time, using the learned equilibrium functions under the NIF design with equity penalty. Specifically, we obtain load-generation profiles by randomly perturbing (5\%) the consumption data used to learn these curves. This can be interpreted as having the data from the data set prescribing a day-ahead forecast, whereas their random perturbation acts as the true realization of the load-generation scenarios. We test the control performance during 12:00 to 16:00. With 100 iterations of~\eqref{eq:power_update} for each minute-based profile, Fig.~\ref{fig:volt_dev} and Fig.~\ref{fig:opt_gap} respectively show the maximum voltage deviation and optimality gap for both NIF and the linear control cases along the evolution. It can be observed that the NIF control significantly enhances the control performance while ensures stability.

\begin{figure}[t]
    \centering
     \subfigure[$\epsilon = 0.1$]{
    \includegraphics[width=.24\textwidth]{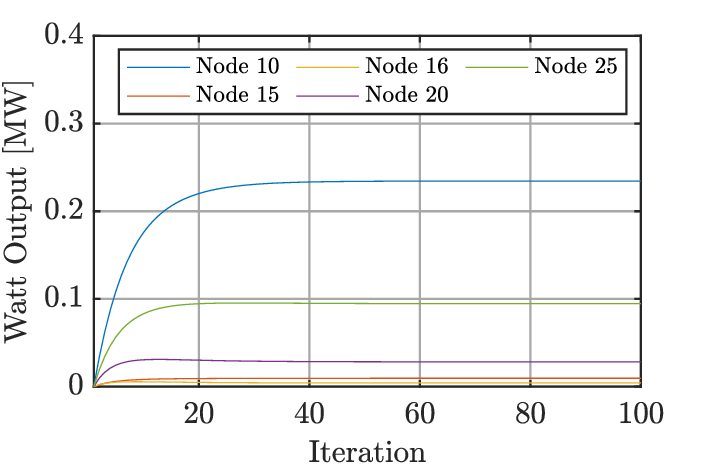}
    \hspace{-2.5ex}
    \includegraphics[width=.24\textwidth]{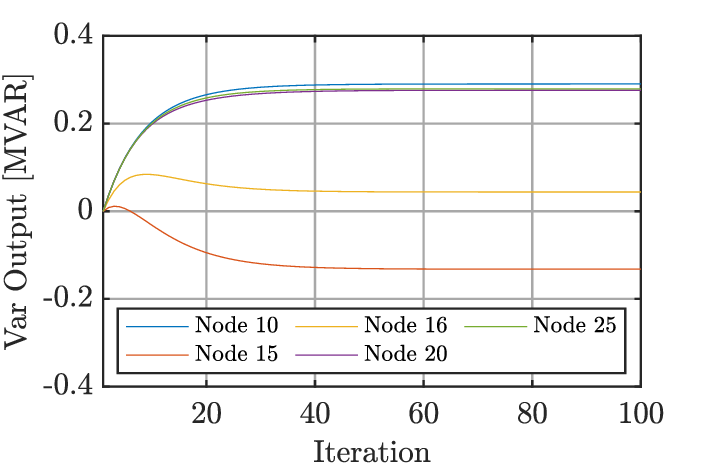}
    \label{fig:stability-0.1}
    }\\
    \vspace{-2ex}
    \subfigure[$\epsilon = 1$]{
    \includegraphics[width=.24\textwidth]{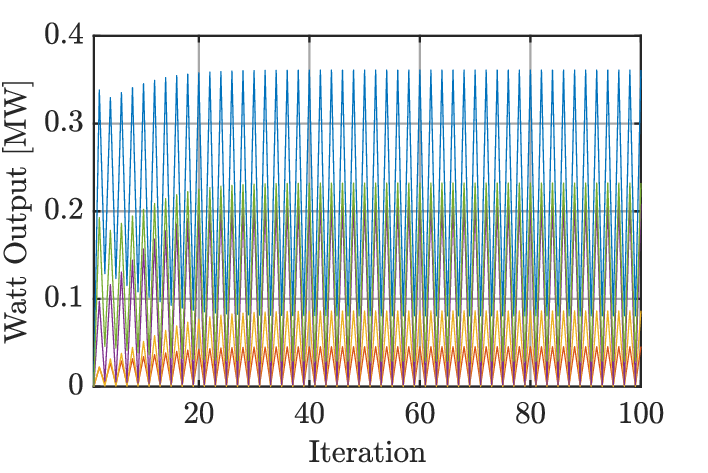}
    \hspace{-2.5ex}
    \includegraphics[width=.24\textwidth]{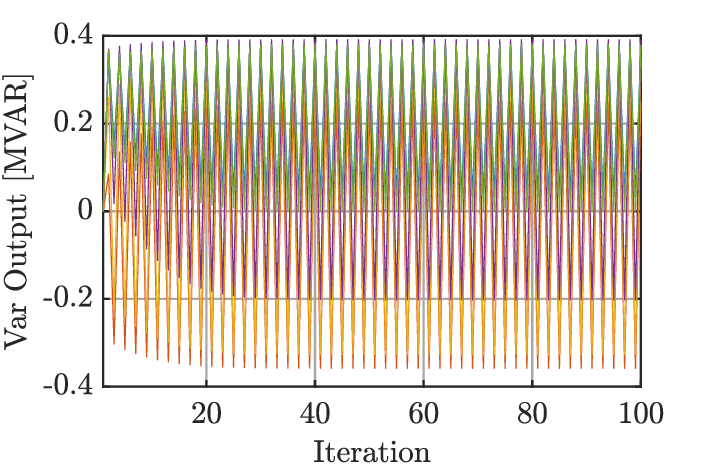}
    \label{fig:stability-1}
    }
  \caption{Evolution of active and reactive power setpoints under the proposed power update rule~\eqref{eq:power_update} with (a) $\epsilon = 0.1$ and (b) $\epsilon = 1$, where we use the power data profiles of the 1095-th minute and consider 100 iterations.
  }
  \label{fig:stability}
\end{figure}

\begin{figure}[tb]
    \centering
    \subfigure[Maximum voltage deviation]{
    \includegraphics[width=.24\textwidth]{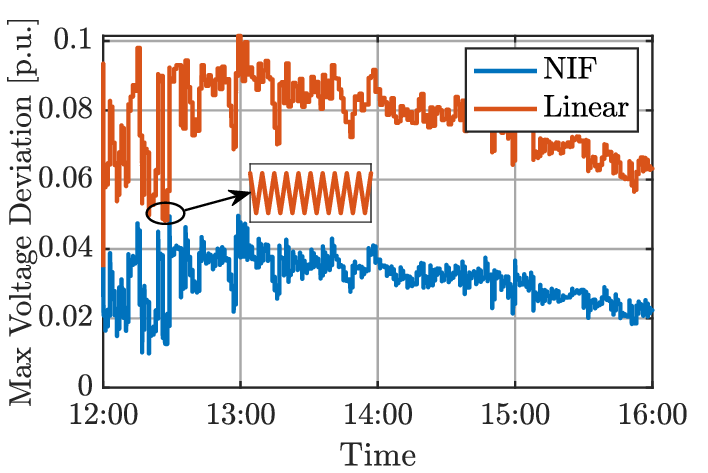}
    \label{fig:volt_dev}
    }\hspace{-3ex}
    \subfigure[Optimality gap]{
    \includegraphics[width=.24\textwidth]{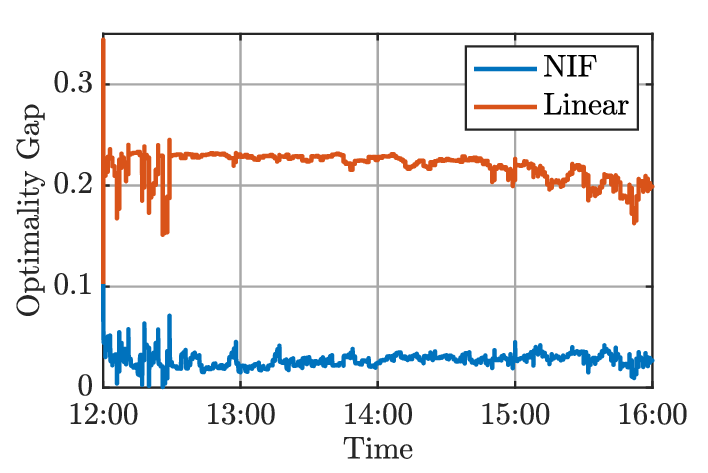}
    \label{fig:opt_gap}
    }
    \caption{Maximum voltage deviation and optimality gap with respect to the OPF solutions along the evolution. The linear control method induces instability between 12:00 and 13:00.}
\end{figure}

Fig.~\ref{fig:curtailment} compares the curtailment evolution with and without equity penalty. We note that node 25 is farther than  node 10 in the electrical distance sense. It can be observed in Fig.~\ref{fig:curtailment_no_regu} and Fig.~\ref{fig:curtailment_regu} that the equity penalty promotes node 10 to curtail more generation so that nodes 10 and 25 contribute more equally to address over-voltage (resulted by relatively high uncontrollable solar generations during this time period). Finally, Fig.~\ref{fig:evolution_penalty} shows the evolution of equity cost value with and and without the equity-promoting penalty, illustrating the overall effectiveness of the equity-promoting design.

\begin{figure*}[htb]
    \centering
    \subfigure[Curtailment without equity design]{
    \includegraphics[width=.26\textwidth]{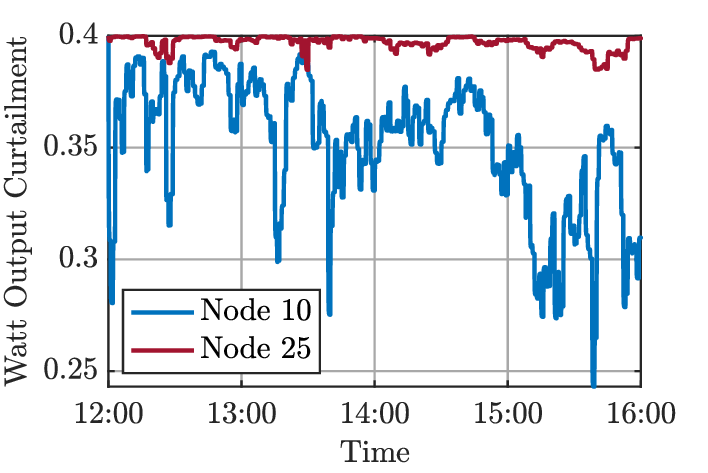}
    \label{fig:curtailment_no_regu}
    }
    \subfigure[Curtailment with equity design]{
    \includegraphics[width=.26\textwidth]{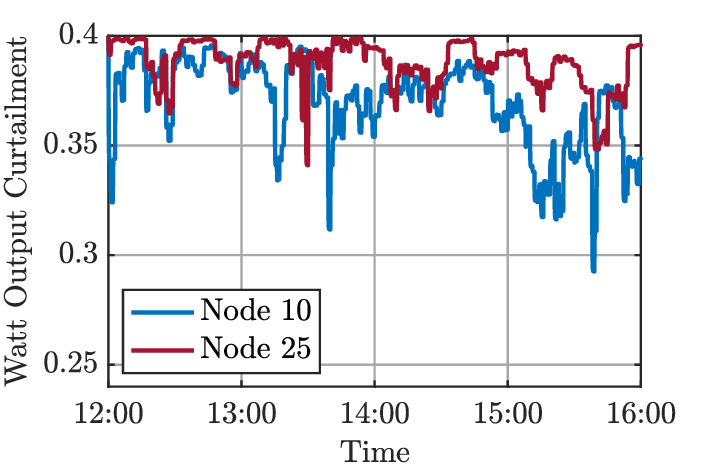}
    \label{fig:curtailment_regu}
    }
    \subfigure[Evolution of equity cost]{
    \includegraphics[width=.26\textwidth]{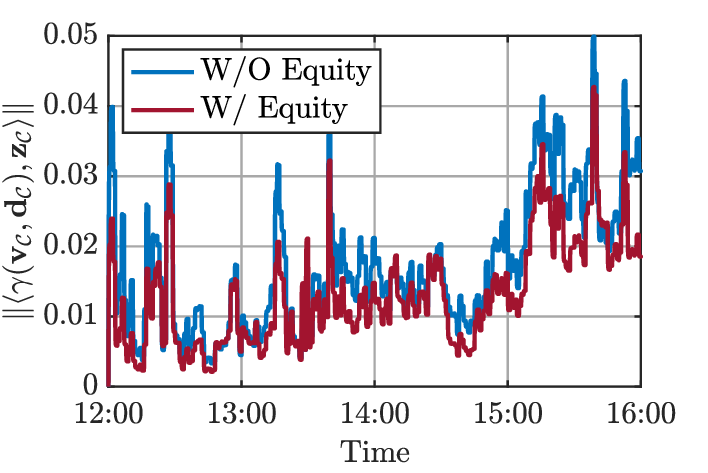}
    \label{fig:evolution_penalty}
    }
  \caption{Comparison of the evolution of curtailment between node 10 and node 25 (a) without and (b) with equity-promoting design and (c) the evolution of the equity cost ($\lambda = 0.0154$).
  }
  \label{fig:curtailment}
\end{figure*} 

    
  

\section{Conclusions}

We have presented a fairness-aware unsupervised learning framework for designing local controllers for DERs in DNs.
The framework aims to learn local surrogates that map local information  to active and reactive power, effectively approximating the solutions derived from OPF models.
We have derived explicit conditions on these surrogates such that the power setpoints converge in a globally asymptotic sense when applied within an incremental control algorithm.
We have shown in power flow simulations that the proposed framework guarantees the voltage stability and significantly enhances the control performance compared to prevalent local control approaches. The simulations have also underscored the pivotal role played by equity-penalty design in fostering fairness in renewable energy curtailment strategies.
Future work will attempt to drop the constraint that equilibrium functions are non-increasing w.r.t. voltage magnitude, which has potential advantages in DNs where DERs have relatively small control capability~\cite{yuan2023constraints}, characterize the performance benefits of the approach proposed here with respect to supervised learning in terms of robustness under changes in network topology, and evaluate the performance of the proposed method with different fairness metrics, robustness against model inaccuracies and uncertainties, as well as its scalability properties in the presence of more DERs.

\appendix[Proof of Theorem~\ref{thm:NIF_stable}]
\label{app:proofs}

From~\eqref{eq:dyn_sys}, the voltage magnitudes evolution is
\begin{equation*}
\vbf_\Cc(t+1) = (1 - \epsilon) \vbf_\Cc(t) + \epsilon 
\begin{bmatrix}
\Rbf \ \Xbf
\end{bmatrix}
\begin{bmatrix}
\gammab(\vbf_\Cc(t),\dbf_\Cc) \\ \xib(\vbf_\Cc(t),\dbf_\Cc)
\end{bmatrix}
+ \epsilon  \hat{\bf v}_\Cc.
\end{equation*}
Define the diagonal matrices $\Pbf(t)$ and $\Qbf(t)$ with elements
\begin{align*}
P_n(t) &:= 
\begin{cases}
\frac{\gamma_{n}(v_n(t),d_n) - \gamma_{n}(v_n^\prime(t),d_n)}{v_n(t) - v_n^\prime(t)} & v_n \neq v_n^\prime, \\
\hfil 0 & v_n = v_n^\prime,
\end{cases}\\
Q_n(t) &:= 
\begin{cases}
\frac{\xi_{n}(v_n(t),d_n) - \xi_{n}(v_n^\prime(t),d_n)}{v_n(t) - v_n^\prime(t)} & v_n \neq v_n', \\
\hfil 0 & v_n = v_n^\prime.
\end{cases}
\end{align*}
For any $\vbf_\Cc(t),\vbf_\Cc^\prime(t) \in \mathbb{R}^{C}$, it holds that
\begin{align*}
    \vbf_\Cc(t &+1) -\vbf_\Cc^\prime(t+1) = (1-\epsilon)(\vbf_\Cc(t) - \vbf_\Cc^\prime(t)) +\notag\\
    & \; \epsilon 
\begin{bmatrix}
    \Rbf \ \Xbf
\end{bmatrix}
\left(
\begin{bmatrix}
\gammab(\vbf_\Cc(t),\dbf_\Cc) \\ \xib(\vbf_\Cc(t),\dbf_\Cc)
\end{bmatrix}
-
\begin{bmatrix}
\gammab(\vbf_\Cc^\prime(t),\dbf_\Cc) \\ \xib(\vbf^\prime_\Cc(t),\dbf_\Cc)
\end{bmatrix}
\right) \\
    & = \big[ (1 - \epsilon)\Ibf + \epsilon \Rbf \Pbf(t) + \epsilon \Xbf \Qbf(t)  \big](\vbf_\Cc(t) - \vbf^\prime_\Cc(t)) \\
    & := \Sbf(t)(\vbf_\Cc(t) - \vbf_\Cc^\prime(t)).
\end{align*}
Consequently, 
\begin{align*}
\lim_{t \rightarrow \infty} \vbf_\Cc(t) - \vbf_\Cc^\prime(t) = \lim_{t \rightarrow \infty} \left[ \prod_{k=0}^{t-1} \Sbf(k) \right] (\vbf_\Cc(0) - \vbf_\Cc^\prime(0)).
\end{align*}
To prove the system stability, we show that the linear map defined by each $\Sbf(k)$ is a contraction, from which we conclude
\begin{equation}
\lim_{t \rightarrow \infty} \left[ \prod_{k=0}^{t-1} \Sbf(k) \right] = 0    
\label{eq:limt}
\end{equation}
and hence the uniqueness of equilibrium and the global asymptotic stability. 
%

Under the NIF design, leveraging condition~\eqref{eq:nonincr_cond} and the fact that the entries of $\Rbf$ and $\Xbf$ are positive~\cite{Cavraro2019Inverter}, we have
\begin{align*}
&\Sbf(k) = \big[ (1 - \epsilon)\Ibf - \epsilon \Rbf |\Pbf(k)| - \epsilon \Xbf |\Qbf(k)|  \big] \\
& = \big[ (1 - \epsilon)\Ibf - \epsilon  \Rbf ( |\Pbf(k)| + \alpha^* |\Qbf(k)|) - \epsilon \hat \Xbf |\Qbf(k)|  \big] \\
 & =\Rbf^{\frac 1 2} \Big[ \underbrace{(1 - \epsilon)\Ibf - \epsilon \Rbf^{\frac 1 2} ( |\Pbf(k)| + \alpha^* |\Qbf(k)|) \Rbf^{\frac 1 2}}_{\Zbf_1(k)} \\
 &\qquad \qquad \qquad \qquad \qquad \qquad - \underbrace{\epsilon \Rbf^{-\frac 1 2} \hat \Xbf |\Qbf(k)| \Rbf^{\frac 1 2}}_{\Zbf_2(k)} \Big] \Rbf^{-\frac 1 2}.
\end{align*}
Since $\Sbf(k)$ is similar to $\Zbf_1(k)-\Zbf_2(k)$, it suffices to show that the latter is a contraction. Noting
$    \left\| \Zbf_1(k) - \Zbf_2(k) \right\| \leq \left\| \Zbf_1(k) \right\| + \left\| \Zbf_2(k) \right\|$,
%
%
we show that $\left\| \Zbf_1(k) \right\| + \left\| \Zbf_2(k) \right\| < 1$.
%
%
Notice
$$
\|\Zbf_2(k)\| \leq \epsilon \|\Rbf^{-\frac 1 2}\| \|\hat \Xbf\| \|\Qbf(k)\| \|\Rbf^{\frac 1 2}\| \leq \epsilon \kappa(\Rbf^{\frac 1 2}) \|\hat \Xbf\| L_q.
$$
Therefore, the condition
\begin{align}\label{eq:con:z1}
    \|\Zbf_1(k)\| &< 1 - \epsilon \kappa(\Rbf^{\frac 1 2}) \|\hat \Xbf\| L_q
\end{align}
is enough to ensure that $\left\| \Zbf_1(k) \right\| + \left\| \Zbf_2(k) \right\| < 1$.
Since $\Zbf_1(k)$ is symmetric, it follows that 
\begin{align*}
\|\Zbf_1(k)\| &= \max \big\{ |1  \!-\! \epsilon \!-\! \epsilon \lambda_{\max}(\Rbf^{\frac 1 2}(|\Pbf(k)| \!+\! \alpha^* |\Qbf(k)|)\Rbf^{\frac 1 2})|, \notag \\
 &|1 \!-\! \epsilon \!-\! \epsilon\lambda_{\min}(\Rbf^{\frac 1 2}(|\Pbf(k)| \!+\! \alpha^* |\Qbf(k)|)\Rbf^{\frac 1 2})| \big \}.
\end{align*}
Since $\Rbf^{\frac 1 2}(|\Pbf(k)| \!+\! \alpha^* |\Qbf(k)|)\Rbf^{\frac 1 2}$ is positive semidefinite,~\eqref{eq:con:z1} holds true if and only if
\begin{subequations}\label{eq:condition-epsilon}
\begin{align}
   & 1 - \epsilon - \epsilon \lambda_{\min}(\Rbf^{\frac 1 2}(|\Pbf(k)| \!+\! \alpha^* |\Qbf(k)|)\Rbf^{\frac 1 2}) \notag \\
   &\qquad \qquad \qquad \qquad \qquad \qquad < 1 - \epsilon \kappa(\Rbf^{\frac 1 2}) \|\hat\Xbf\| L_q, \label{eq:condition-epsilon-1}\\
   & 1 - \epsilon - \epsilon \lambda_{\max}(\Rbf^{\frac 1 2}(|\Pbf(k)| \!+\! \alpha^* |\Qbf(k)|)\Rbf^{\frac 1 2}) \notag \\
   &\qquad \qquad \qquad \qquad \qquad \qquad > \epsilon \kappa(\Rbf^{\frac 1 2}) \|\hat\Xbf\| L_q - 1. \label{eq:condition-epsilon-2}
\end{align}    
\end{subequations}
To prove inequality~\eqref{eq:condition-epsilon-1}, heed that condition~\eqref{eq:nonincr_cond_a} implies that
$$1 - \epsilon \kappa(\Rbf^{\frac 1 2}) \|\hat\Xbf\| L_q > 1-\epsilon.$$
Then,~\eqref{eq:condition-epsilon-1} holds as $\lambda_{\min}(\Rbf^{\frac 1 2}(|\Pbf(k)| + \alpha^* |\Qbf(k)|)\Rbf^{\frac 1 2}) \geq 0$. 
Inequality~\eqref{eq:condition-epsilon-2} follows from~\eqref{eq:nonincr_cond_b} and from the fact that $\epsilon \in [0,1]$ and $\Rbf^{\frac 1 2}(|\Pbf(k)| \!+\! \alpha^* |\Qbf(k)|)\Rbf^{\frac 1 2}$ is similar to $\Rbf(|\Pbf(k)| \!+\! \alpha^* |\Qbf(k)|)$. This ends the proof.

\bibliographystyle{ieeetr}

\end{document}